\documentclass[a4paper]{article}
\usepackage[utf8]{inputenc}
\usepackage[T1]{fontenc}

\usepackage{amsmath}
\usepackage{amssymb}
\usepackage{amsthm}

\usepackage{array,multirow,graphicx}
\usepackage{tabularx}
\usepackage{url}
\usepackage{tikz}

\usetikzlibrary{backgrounds,fit,calc,chains,matrix,positioning,shapes}
\tikzset{>=stealth}

\usepackage{pgfplots}
\usepackage{subfig}

\usepackage{dsfont}

\date{}

\let\oldmaketitle=\maketitle
\renewcommand{\maketitle}{\oldmaketitle\thispagestyle{empty}}

\usepackage{subfig}

\newcommand{\sn}{\smallskip\noindent}

\newcommand{\bbbb}{\ensuremath{\mathrm{I\!B}}}

\newcommand{\bfunc}{\ensuremath{\mathcal{B}}}

\newtheorem{definition}{Definition}
\newtheorem{lemma}{Lemma}
\newtheorem{example}{Example}

\newtheorem{theorem}{Theorem}
\newtheorem{corollary}{Corollary}
\newcommand{\eqdef}{\stackrel{\text{\normalfont def}}{=}}

\newcommand{\T}{\ensuremath{\mathrm{T}}}

\def\hang{\hangindent19pt}
\def\d@anger{\medbreak\begingroup\clubpenalty=10000
 \def\par{\endgraf\endgroup\medbreak} \noindent\hang\hangafter=-2
 \hbox to0pt{\hskip-\hangindent\dbend\hfill}\small}
\outer\def\danger{\d@anger}

\usepackage{tikz}
\usetikzlibrary{shapes,arrows}
\usetikzlibrary{backgrounds,calc,decorations.pathreplacing,positioning}

\tikzset{%
  task/.style={draw,inner sep=3pt,align=center,fill=RAgray},
  cucumber task/.style={task},
  word/.style={inner sep=0pt,outer sep=0pt,font=\footnotesize,text height=1.5ex,text depth=.25ex},
  highlight word/.style={word,fill=RAgray,draw=black!20,rounded corners=2pt},
  class/.style={draw,rectangle split,rectangle split parts=3,rectangle split part align={center,left},align=left,font=\sffamily\scriptsize\color{white},inner sep=2.5pt,fill=RAblue}
}

\tikzset{%
  font=\footnotesize,>=latex,
  block1/.style={draw,rectangle,dashed,minimum height=4em,minimum width=12em},
  block2/.style={draw,rectangle,minimum height=30em,minimum width=10em},
  block3/.style={draw,rectangle,minimum height=30em,minimum width=10em,blue},
  block4/.style={draw,rectangle,minimum height=2.4em,minimum width=1.4em,blue},
  block5/.style={draw,rectangle,minimum height=4em,minimum width=2.8em,blue},
  block6/.style={draw,rectangle,minimum height=4em,minimum width=2em,red},
  block7/.style={draw,rectangle,minimum height=3.7em,minimum width=2.5em,blue},
  block8/.style={draw,rectangle,minimum height=3.7em,minimum width=1.8em,red},
  block9/.style={draw,rectangle,minimum height=2.8em,minimum width=3.5em,blue},
  block10/.style={draw,rectangle,minimum height=2.8em,minimum width=2.8em,blue},
  block11/.style={draw,rectangle,minimum height=1em,minimum width=20em},
  vertex/.style={draw,circle,inner sep=2pt},
  decision/.style = {diamond, draw, text width=4em, text  centered, node distance=3cm, inner sep=0pt},
  block/.style = {rectangle, draw,,minimum height=50em,width=15em, text centered, rounded corners, minimum height=1em},
   line/.style = {draw, ->},
     line1/.style = {draw},
  cloud/.style = {draw, ellipse,fill=red!20, node distance=3cm, minimum height=2em}
}
\def \tikzVshape{

  \begin{tikzpicture}[scale=0.8]
    \draw[line width=0.300000] (0.500000,1.250000) -- (4.500000,1.250000);
    \draw (0.400000,1.250000) node [left] {$x_1$};
    \draw (4.600000,1.250000) node [right] {$y_1$};
    \draw[line width=0.300000] (0.500000,0.750000) -- (4.500000,0.750000);
    \draw (0.400000,0.750000) node [left] {$x_2$};
    \draw (4.600000,0.750000) node [right] {$y_2$};
      \draw (0.400000,0.50000) node [left] {$.$};
    \draw (4.600000,0.50000) node [right] {$.$};
    \draw (0.400000,0.30000) node [left] {$.$};
    \draw (4.600000,0.30000) node [right] {$.$};
    \draw (0.400000,0.40000) node [left] {$.$};
    \draw (4.600000,0.40000) node [right] {$.$};
    \draw[line width=0.300000] (0.500000,0.00000) -- (4.500000,0.00000);
    \draw (0.400000,0.00000) node [left] {$x_n$};
    \draw (4.600000,0.00000) node [right] {$y_n$};
 
     \draw[line width=0.300000] (0.750000,0.250000) -- (0.750000,1.250000);
      \draw[line width=0.300000] (4.250000,0.250000) -- (4.250000,1.250000);
      \draw[fill=white,rounded corners=2pt] (0.65,-0.1) rectangle ++(0.2,1.1);
      \draw[line width=0.3] (0.75,1.25) circle (0.2) (0.75,1.05) -- (0.75,1.45);
    \draw[fill=white,rounded corners=2pt] (4.15,-0.1) rectangle ++(0.2,1.1);
      \draw[line width=0.3] (4.25,1.25) circle (0.2) (4.25,1.05) -- (4.25,1.45);
  \node[below, rotate=90, align=left] at (0.45,-0.8) {\footnotesize $g_1$};
      \node[below, rotate=90, align=left] at (3.95,-0.8) {\footnotesize $g_{2n-1}$};
  
  \draw[line width=0.300000] (0.500000,1.250000) -- (4.500000,1.250000);
   
      \draw[line width=0.300000] (1.000000,0.00000) -- (4.000000,00000);
   
      \draw[line width=0.300000] (1.50000,0.250000) -- (1.50000,1.250000);
      \draw[fill=white,rounded corners=2pt] (1.4,-0.1) rectangle ++(0.2,0.6);
      \draw[fill=white,rounded corners=2pt] (1.4,1) rectangle ++(0.2,0.35);
      \draw[line width=0.3] (1.5,0.75) circle (0.2) (1.5,.55) -- (1.5,.95);

      \draw[line width=0.300000] (3.50000,0.250000) -- (3.50000,1.250000);
      \draw[fill=white,rounded corners=2pt] (3.4,-0.1) rectangle ++(0.2,0.6);
      \draw[fill=white,rounded corners=2pt] (3.4,1) rectangle ++(0.2,0.35);
      \draw[line width=0.3] (3.5,0.75) circle (0.2) (3.5,.55) -- (3.5,.95);

      \draw[line width=0.300000] (1.000000,0.750000) -- (4.000000,0.750000);
  \node[below, rotate=90, align=left] at (1.2,-0.8) {\footnotesize $g_2$};
      \node[below, rotate=90, align=left] at (3.2,-0.8) {\footnotesize $g_{2n-2}$};   
    
    \draw[line width=0.300000] (2.50000,0.00000) -- (2.50000,1.250000);
    \draw[fill=white,rounded corners=2pt] (2.4,0.25) rectangle ++(0.2,1.1);
    \draw[line width=0.3] (2.5,0.00) circle (0.2) (2.5,-0.20) -- (2.5,0.2);
      \node[below, rotate=90, align=left] at (2.2,-0.8) {\footnotesize $g_{n}$};
\node[below] at (2.1,-0.8) {\footnotesize $\cdots$};
      \node[below] at (3.1,-0.8) {\footnotesize $\cdots$};       
  	\draw (2.40000,.30000) node [left] {\footnotesize $\cdots$};
    \draw (2.600000,.30000) node [right] {\footnotesize $\cdots$};

  \end{tikzpicture}

}
\def \tikzmct{%
    \begin{tikzpicture}[scale=.7]
           \draw[line width=0.300000] (0.500000,2.250000) -- (1.000000,2.250000);
    \draw (0.400000,2.250000) node [left] {$x_1$};
    \draw (1.100000,2.250000) node [right] {$x_1$};
    \draw[line width=0.300000] (0.500000,1.750000) -- (1.000000,1.750000);
    \draw (0.400000,1.750000) node [left] {$x_2$};
    \draw (1.100000,1.750000) node [right] {$x_2$};
  \draw (0.400000,1.150000) node [left] {$.$};
    \draw (1.100000,1.150000) node [right] {$.$};
    \draw (0.400000,1.250000) node [left] {$.$};
    \draw (1.100000,1.250000) node [right] {$.$};
     \draw (0.400000,1.350000) node [left] {$.$};
    \draw (1.100000,1.350000) node [right] {$.$};
    \draw[line width=0.300000] (0.500000,0.750000) -- (1.000000,0.750000);
    \draw (0.400000,0.750000) node [left] {$x_{n-1}$};
    \draw (1.100000,0.750000) node [right] {$x_{n-1}$};
    \draw[line width=0.300000] (0.500000,0.250000) -- (1.000000,0.250000);
    \draw (0.400000,0.250000) node [left] {$x_n$};
    \draw (1.100000,0.250000) node [right] {$x_n\oplus \left(x_1\land\dots\land x_{n-1}\right)$};
    \draw[line width=0.300000] (0.750000,0.250000) -- (0.750000,2.250000);
    \draw[fill] (0.750000,2.250000) circle (0.100000);
    \draw[fill] (0.750000,0.750000) circle (0.100000);
    \draw[line width=0.3] (0.75,0.25) circle (0.2) (0.75,0.05) --(0.75,0.45);
    \end{tikzpicture}}
    
\def \tikzmpmct{%
    \begin{tikzpicture}[scale=.7]
      \draw[line width=0.300000] (0.500000,2.250000) -- (1.000000,2.250000);
    \draw (0.400000,2.250000) node [left] {$x_1$};
    \draw (1.100000,2.250000) node [right] {$x_1$};
    \draw[line width=0.300000] (0.500000,1.750000) -- (1.000000,1.750000);
    \draw (0.400000,1.750000) node [left] {$x_2$};
    \draw (1.100000,1.750000) node [right] {$x_2$};
  \draw (0.400000,1.150000) node [left] {$.$};
    \draw (1.100000,1.150000) node [right] {$.$};
    \draw (0.400000,1.250000) node [left] {$.$};
    \draw (1.100000,1.250000) node [right] {$.$};
     \draw (0.400000,1.350000) node [left] {$.$};
    \draw (1.100000,1.350000) node [right] {$.$};
    \draw[line width=0.300000] (0.500000,0.750000) -- (1.000000,0.750000);
    \draw (0.400000,0.750000) node [left] {$x_{n-1}$};
    \draw (1.100000,0.750000) node [right] {$x_{n-1}$};
    \draw[line width=0.300000] (0.500000,0.250000) -- (1.000000,0.250000);
    \draw (0.400000,0.250000) node [left] {$x_n$};
    \draw (1.100000,0.250000) node [right] {$x_n\oplus\left(\bar x_1\land\dots\land x_{n-1}\right)$};
    \draw[line width=0.300000] (0.750000,0.250000) -- (0.750000,2.250000);
    \draw[fill=white] (0.750000,2.250000) circle (0.100000);
    \draw[fill] (0.750000,0.750000) circle (0.100000);
    \draw[line width=0.3] (0.75,0.25) circle (0.2) (0.75,0.05) --(0.75,0.45);
      \end{tikzpicture}}

\def\tikzSingleTarget{%
  \begin{tikzpicture}[scale=.7]
    \draw[line width=0.3] (0.25,2.25) -- (1.25,2.25);
    \draw[line width=0.3] (0.25,1.75) -- (1.25,1.75);
    \draw[line width=0.3] (0.25,0.75) -- (1.25,0.75);
    \draw[line width=0.3] (0.25,0.25) -- (1.25,0.25);

      \draw (0.3,2.25) node [left] {$x_1$};
      \draw (0.3,1.75) node [left] {$x_2$};
      \draw (0.3,1.35) node [left] {$\vdots$};
      \draw (0.3,0.75) node [left] {$x_{n-1}$};
      \draw (0.3,0.25) node [left] {$x_n$};

      \draw (1.2,2.25) node [right] {$x_1$};
      \draw (1.2,1.75) node [right] {$x_2$};
      \draw (1.2,1.35) node [right] {$\vdots$};
      \draw (1.2,0.75) node [right] {$x_{n-1}$};
      \draw (1.2,0.25) node [right] {$x_n\oplus g(x_1,x_2,\dots,x_{n-1})$};

    \draw[line width=0.3] (0.75,0.25) -- (0.75,0.65);

    \draw[rounded corners=2pt,fill=white] (0.55,2.35) coordinate (a) rectangle (.95,0.65) coordinate (b);
    \node at ($(a)!.5!(b)$) {$g$};

    \draw[line width=0.3] (0.75,0.25) circle (0.2) (0.75,0.05) -- (0.75,0.45);

  \end{tikzpicture}}
\def \tikzrevcircuit{%
\begin{tikzpicture}[scale=.8]
\draw[line width=0.300000] (0.500000,1.250000) -- (4.00000,1.250000);
\draw (0.400000,1.250000) node [left] {$x_1 = 0$};
\draw (4.100000,1.250000) node [right] {$y_1 =1$};
\draw[line width=0.300000] (0.500000,0.750000) -- (4.00000,0.750000);
\draw (0.400000,0.750000) node [left] {$x_2 = 1$};
\draw (4.100000,0.750000) node [right] {$y_2 = 1$};
\draw[line width=0.300000] (0.500000,0.250000) -- (4.00000,0.250000);
\draw (0.400000,0.250000) node [left] {$x_3 = 0$};
\draw (4.100000,0.250000) node [right] {$y_3 = 1$};
\draw[line width=0.300000] (0.750000,1.250000) -- (0.750000,1.250000);
\draw[line width=0.3] (0.75,1.25) circle (0.2) (0.75,1.05) -- (0.75,1.45);
\draw (1.100000,-0.100000) node [left] {$g_1$};
\draw (1.500000,1.380000) node [left] {$1$};
\draw (1.500000,0.880000) node [left] {$1$};
\draw (1.500000,0.380000) node [left] {$0$};
\draw[line width=0.300000] (1.750000,0.250000) -- (1.750000,1.250000);
\draw[fill] (1.750000,1.250000) circle (0.100000);
\draw[fill] (1.750000,0.750000) circle (0.100000);
\draw[line width=0.3] (1.75,0.25) circle (0.2) (1.75,0.05) -- (1.75,0.45);
\draw (2.100000,-0.100000) node [left] {$g_2$};
\draw (2.500000,1.380000) node [left] {$1$};
\draw (2.500000,0.880000) node [left] {$1$};
\draw (2.500000,0.380000) node [left] {$1$};
\draw[line width=0.300000] (2.750000,0.750000) -- (2.750000,1.250000);
\draw[fill] (2.750000,1.250000) circle (0.100000);
\draw[line width=0.3] (2.75,0.75) circle (0.2) (2.75,0.55) -- (2.75,0.95);
\draw (3.100000,-0.100000) node [left] {$g_3$};
\draw (3.500000,1.380000) node [left] {$1$};
\draw (3.500000,0.880000) node [left] {$0$};
\draw (3.500000,0.380000) node [left] {$1$};
\draw[line width=0.300000] (3.750000,0.250000) -- (3.750000,0.750000);
\draw[fill] (3.750000,0.250000) circle (0.100000);
\draw[line width=0.3] (3.75,0.75) circle (0.2) (3.75,0.55) -- (3.75,0.95);
\draw (4.100000,-0.100000) node [left] {$g_4$};
\end{tikzpicture}}

\def \tikzdiagram{%
\begin{tikzpicture}[auto,>=latex,every text node part/.style={align=center}]]
      
\node [name=gc]{
\begin{tikzpicture}[scale=.8]
\draw[line width=0.300000] (0.500000,1.250000) -- (4.500000,1.250000);
    \draw (0.400000,1.250000) node [left] {$x_1$};
    \draw (4.600000,1.250000) node [right] {$y_1$};
    \draw[line width=0.300000] (0.500000,0.750000) -- (4.500000,0.750000);
    \draw (0.400000,0.750000) node [left] {$x_2$};
    \draw (4.600000,0.750000) node [right] {$y_2$};
      \draw (0.400000,0.50000) node [left] {$.$};
    \draw (4.600000,0.50000) node [right] {$.$};
    \draw (0.400000,0.30000) node [left] {$.$};
    \draw (4.600000,0.30000) node [right] {$.$};
    \draw (0.400000,0.40000) node [left] {$.$};
    \draw (4.600000,0.40000) node [right] {$.$};
    \draw[line width=0.300000] (0.500000,0.00000) -- (4.500000,0.00000);
    \draw (0.400000,0.00000) node [left] {$x_n$};
    \draw (4.600000,0.00000) node [right] {$y_n$};
 
     \draw[line width=0.300000] (0.750000,0.250000) -- (0.750000,1.250000);
      \draw[line width=0.300000] (4.250000,0.250000) -- (4.250000,1.250000);
      \draw[fill=white,rounded corners=2pt] (0.65,-0.1) rectangle ++(0.2,1.1);
      \draw[line width=0.3] (0.75,1.25) circle (0.2) (0.75,1.05) -- (0.75,1.45);
    \draw[fill=white,rounded corners=2pt] (4.15,-0.1) rectangle ++(0.2,1.1);
      \draw[line width=0.3] (4.25,1.25) circle (0.2) (4.25,1.05) -- (4.25,1.45);
  
  \draw[line width=0.300000] (0.500000,1.250000) -- (4.500000,1.250000);
   
      \draw[line width=0.300000] (1.000000,0.00000) -- (4.000000,00000);
   
      \draw[line width=0.300000] (1.50000,0.250000) -- (1.50000,1.250000);
      \draw[fill=white,rounded corners=2pt] (1.4,-0.1) rectangle ++(0.2,0.6);
      \draw[fill=white,rounded corners=2pt] (1.4,1) rectangle ++(0.2,0.35);
      \draw[line width=0.3] (1.5,0.75) circle (0.2) (1.5,.55) -- (1.5,.95);

      \draw[line width=0.300000] (3.50000,0.250000) -- (3.50000,1.250000);
      \draw[fill=white,rounded corners=2pt] (3.4,-0.1) rectangle ++(0.2,0.6);
      \draw[fill=white,rounded corners=2pt] (3.4,1) rectangle ++(0.2,0.35);
      \draw[line width=0.3] (3.5,0.75) circle (0.2) (3.5,.55) -- (3.5,.95);

      \draw[line width=0.300000] (1.000000,0.750000) -- (4.000000,0.750000);
    
    \draw[line width=0.300000] (2.50000,0.00000) -- (2.50000,1.250000);
    \draw[fill=white,rounded corners=2pt] (2.4,0.25) rectangle ++(0.2,1.1);
    \draw[line width=0.3] (2.5,0.00) circle (0.2) (2.5,-0.20) -- (2.5,0.2);
  \end{tikzpicture}};
\node [below of=gc,node distance=1.1cm](tgn){Linear Complexity
};
\node [above of=gc,node distance=1.1cm](tgnn){Single-target gates circuit};
\node [right of=gc,node distance=9cm](tc){
\begin{tikzpicture}[scale=.8]
\draw[line width=0.300000] (0.500000,1.250000) -- (4.500000,1.250000);
    \draw (0.400000,1.250000) node [left] {$x_1$};
    \draw (4.600000,1.250000) node [right] {$y_1$};
    \draw[line width=0.300000] (0.500000,0.750000) -- (4.500000,0.750000);
    \draw (0.400000,0.750000) node [left] {$x_2$};
    \draw (4.600000,0.750000) node [right] {$y_2$};
  \draw (0.400000,0.50000) node [left] {$.$};
    \draw (4.600000,0.50000) node [right] {$.$};
    \draw (0.400000,0.30000) node [left] {$.$};
    \draw (4.600000,0.30000) node [right] {$.$};
    \draw (0.400000,0.40000) node [left] {$.$};
    \draw (4.600000,0.40000) node [right] {$.$};
    \draw[line width=0.300000] (0.500000,0.00000) -- (4.500000,0.00000);
    \draw (0.400000,0.00000) node [left] {$x_n$};
    \draw (4.600000,0.00000) node [right] {$y_n$};
\draw[fill=white,opacity=.8] (0.6,-0.2) rectangle ++(3.8,1.65);

\end{tikzpicture}};
\node [below of=tc,node distance=1.1cm](tcn){Exponential Complexity};
\node [above of=tc,node distance=1.1cm](tcnn){Toffoli gates circuit};
\node [below of=tgn,node distance=.7cm](a1){};
\node [right of=a1,node distance=1cm](a3){};
\node [right of=a3,node distance=3.5cm,text width=3cm](a4){Cascades of~ Toffoli \\(ESOP)};
\draw [->] (gc) -- node[node distance=0.5cm] {Mapping} (tc);  
\draw [->] (a4) -| node[node distance=0.5cm] {} (tcn);
\draw  (a4) -| node[node distance=0.5cm] {} (tgn);
\end{tikzpicture}}

\def \tikzGeneralGate_1{%
\begin{tikzpicture}[scale=.8]
\draw[line width=0.300000] (0.500000,2.750000) -- (2.500000,2.75000);
\draw (0.400000,2.750000) node [left] {$x_1$};
\draw (2.600000,2.750000) node [right] {$x_1$};
\draw[line width=0.300000] (0.500000,2.250000) -- (2.500000,2.250000);
\draw (0.400000,2.250000) node [left] {$x_2$};
\draw (2.600000,2.250000) node [right] {$x_2$};
\draw (0.400000,1.900000) node [left] {$\vdots\;\;$};
\draw (2.600000,1.900000) node [right] {$\;\vdots$};
\draw[line width=0.300000] (0.500000,1.250000) -- (2.500000,1.250000);
\draw (0.400000,1.250000) node [left] {$x_{n-1}$};
\draw (2.600000,1.250000) node [right] {$x_{n-1}$};
\draw[line width=0.300000] (0.500000,0.75000) -- (2.500000,0.75000);
\draw (0.400000,0.750000) node [left] {$x_{n}$};
\draw (2.600000,0.750000) node [right] {$x_n$};
\draw[line width=0.300000] (0.500000,0.250000) -- (2.500000,0.250000);
\draw (0.400000,0.250000) node [left] {$x_{n+1}$};
\draw (2.600000,0.250000) node [right] {$g_{n+1}$};
\draw[line width=0.3] (0.85,0.75) circle (0.18) (0.85,0.57) -- (0.85,0.935); 
\draw[fill=white,opacity=.8] (1.19,1) rectangle ++(0.25,2.0);
\draw[line width=0.3] (1.3,0.25) circle (0.18) (1.3,0.07) -- (1.3,1); 
\draw[fill] (1.30000,0.750000) circle (0.100000);
\draw[line width=0.3] (1.8,0.75) circle (0.18) (1.8,0.57) -- (1.8,0.935); 
\draw[fill] (2.250000,0.750000) circle (0.100000);
\draw[fill=white,opacity=.8] (2.12,1) rectangle ++(0.25,2.0);
\draw[line width=0.3] (2.25,0.25) circle (0.18) (2.25,0.07) -- (2.25,1); 
\draw[fill] (2.250000,0.750000) circle (0.100000);
\end{tikzpicture}}

\def \tikz_gate_gen{%
\begin{tikzpicture}[scale=.8]
\draw[line width=0.300000] (0.500000,2.750000) -- (2.050000,2.75000);
\draw (0.400000,2.750000) node [left] {$x_1$};
\draw (2.1500000,2.750000) node [right] {$x_1$};
\draw[line width=0.300000] (0.500000,2.250000) -- (2.05000,2.250000);
\draw (0.400000,2.250000) node [left] {$x_2$};
\draw (2.150000,2.250000) node [right] {$x_2$};
\draw[line width=0.300000] (0.500000,1.750000) -- (2.05000,1.750000);
\draw (0.400000,1.750000) node [left] {$x_{n-1}$};
\draw (2.150000,1.750000) node [right] {$x_{n-1}$};
\draw[line width=0.300000] (0.500000,0.75000) -- (2.05000,0.75000);
\draw (0.400000,0.750000) node [left] {$x_{n}$};
\draw (2.150000,0.750000) node [right] {$x_n$};
\draw[line width=0.300000] (0.500000,0.250000) -- (2.05000,0.250000);
\draw (0.400000,0.250000) node [left] {$x_{n+1}$};
\draw (2.150000,0.250000) node [right] {$g_{n+1}$};
\draw[line width=0.300000,dotted] (0.100000,1.60000) -- (0.100000,1.20000);
\draw[line width=0.300000,dotted] (1.800000,1.60000) -- (1.800000,1.30000);
\draw[fill=white,opacity=.8] (0.6,1) rectangle ++(0.25,2.0);
\draw[line width=0.3] (.73,0.25) circle (0.18) (0.73,0.07) -- (0.73,1); 
\draw[fill] (0.73,0.750000) circle (0.100000);
\draw[line width=0.3] (1.28,0.25) circle (0.18) (1.28,0.07) -- (1.28,0.435); 
\draw[fill=white,opacity=.8] (1.7,1) rectangle ++(0.25,2.0);
\draw[line width=0.3] (1.83,0.25) circle (0.18) (1.83,0.07) -- (1.83,1); 
\draw[fill] (1.830000,0.750000) circle (0.100000);
\end{tikzpicture}}

\title{A framework for reversible circuit complexity}
\author{%
  Mathias Soeken$^{1,2}$ \qquad Nabila Abdessaied$^2$ \qquad
  Rolf Drechsler$^{1,2}$\\
  $^1$ Faculty of Mathematics and Computer Science, University of Bremen,
    Germany \\
  $^2$ Cyber-Physical Systems, DFKI GmbH, Bremen, Germany \\
  \{msoeken,nabila,drechsle\}@informatik.uni-bremen.de}

\begin{document}

\maketitle

\begin{abstract}
  Reversible single-target gates are a generalization of Toffoli gates which are
  a helpful formal representation for the description of synthesis algorithms
  but are too general for an actual implementation based on some technology.
  There is an exponential lower bound on the number of Toffoli gates required to
  implement any reversible function, however, there is also a linear upper bound
  on the number of single-target gates which can be proven using a constructive
  proof based on a former presented synthesis algorithm.  Since single-target
  gates can be mapped to a cascade of Toffoli gates, this synthesis algorithm
  provides an interesting framework for reversible circuit complexity.  The
  paper motivates this framework and illustrates first possible applications
  based on it.
\end{abstract}

\section{Introduction}
In this paper we concern ourselves with a special class of Boolean
multiple-output functions called \emph{reversible functions} which are those
functions~$f:\bbbb^n\to\bbbb^m$ that are bijective, i.e.~there exists a
$1$-to-$1$ mapping from the inputs to the outputs.  Clearly, if~$f$ is
reversible, then~$n=m$.  Boolean multiple-output functions that are not
reversible are called irreversible.  Reversible functions can be implemented in
terms of reversible circuits.  Reversible functions and circuits play an
important role in quantum computation~\cite{BBC+95} and low-power
computing~\cite{BAP+12}.

A lot of research has been investigating the complexity of Boolean functions and
Boolean circuits in the past~\cite{Weg87,Vol99}, however, no thorough
considerations have been made for reversible functions and circuits so
far. Recently the complexity of synthesis~\cite{CCC14} and equivalence
checking~\cite{Jor13} have individually been investigated.  Based on two bounds
for the number of gates in reversible circuits, in this paper we propose a
general framework that is helpful for the analysis of reversible circuit
complexity.  The first bound is a linear upper bound with respect to
single-target gates~\cite{VR08}, the second one is an exponential lower bound
with respect to Toffoli gates~\cite{MDM05}.  Single-target gates are convenient
to be used as a model for analysis of reversible functions as well as for the
description of synthesis algorithms~\cite{VR08,NSTD14}, whereas Toffoli gates
have been used in practical implementations~\cite{DV02}.  Single-target gates
can be mapped to a cascade of Toffoli gates using exclusive sum-of-products
(ESOP) mapping.

The paper is structured as follows.  We first review ESOP mapping and reversible
circuits.  Afterwards, we prove the two bounds that are discussed above.
Section~\ref{sec:framework} describes the framework for complexity analysis and
Sect.~\ref{sec:application} illustrates an application based on the framework
for a better than optimal embedding strategy.  The paper concludes in
Sect.~\ref{sec:conclusion}.

\section{Notation and Definitions}
\subsection{Boolean Functions}
Let~$\bbbb\eqdef\{0,1\}$ denote the \emph{Boolean values}.  Then we refer to
$\bfunc_{n,m}\eqdef\{f\mid f\colon\bbbb^n\to\bbbb^m\}$ as the set of all
\emph{Boolean multiple-output functions} with $n$~inputs and $m$~outputs.  There
are $2^{m2^n}$ such Boolean functions.  We write~$\bfunc_n\eqdef\bfunc_{n,1}$
and assume that each $f\in\bfunc_n$ is represented by a propositional formula
over the variables $x_1,\dots,x_n$.  Furthermore, we assume that each
function~$f\in \bfunc_{n,m}$ is represented as a tuple $f=(f_1,\dots,f_m)$
where~$f_i\in\bfunc_n$ for each~$i\in\{1,\dots,m\}$ and hence~$f(\vec
x)=(f_1(\vec x),\dots,f_m(\vec x))$ for each~$\vec x\in\bbbb^n$.

\subsection{Exclusive Sum of Products}
Exclusive sum-of-products~(ESOPs,~\cite{Sas93}) are two-level descriptions for
Boolean functions in which a function is composed of~$k$ \emph{product terms}
that are combined using the exclusive-OR~(EXOR,~$\oplus$) operation.  A product
term is the conjunction of $l_i$ literals where a \emph{literal} is either a
propositional variable~$x^1=x$ or its negation~$x^0=\bar x$.  ESOPs are the most
general form of two-level AND-EXOR expressions:
\begin{equation}
  \label{eq:esop}
  f=\bigoplus_{i=1}^k x^{p_{i_1}}_{i_1}\land\cdots\land x^{p_{i_{l_i}}}_{i_{l_i}}
\end{equation}
Several restricted subclasses have been considered in the past,
e.g.~\emph{positive polarity Reed-Muller expressions}~(PPRM~\cite{Sas93}), in
which all literals are positive.  There are further subclasses and most of them
can be defined based on applying the following decomposition rules. An arbitrary
Boolean function~$f (x_1,x_2,\dots, x_n )$ can be expanded as
\begin{alignat}{2}
  f&=\bar x_i f_{\bar x_i} \oplus x_if_{x_i} \tag{Shannon} \\
  f&=f_{\bar x_i} \oplus x_i(f_{\bar x_i}\oplus f_{x_i}) \tag{positive Davio} \\
  f&=f_{x_i} \oplus \bar x_i(f_{\bar x_i}\oplus f_{x_i}) \tag{negative Davio}
\end{alignat}
with \emph{co-factors}
$f_{\bar x_i}=f(x_1, \dots, x_{i-1}, 0, x_{i+1}, \dots, x_n )$ and
$f_{x_i}=f(x_1, \dots, x_{i-1}, 1, x_{i+1}, \dots, x_n )$.

\subsection{Reversible Circuits}
Reversible functions can be realized by reversible circuits that consist of at
least~$n$ lines and are constructed as cascades of reversible gates that belong
to a certain gate library. The most common gate library consists of Toffoli
gates or single-target gates.
\begin{definition}[Reversible single-target gate]
  Given a set of variables~$X=\{x_1,\dots,x_n\}$, a
  \emph{reversible single-target gate}~$\T_g(C,t)$ with \emph{control
  lines}~$C=\{x_{i_1},\dots,x_{i_k}\}\subset X$, a \emph{target line}~$t\in
  X\setminus C$, and a \emph{control function}~$g\in\bfunc_k$
  inverts the variable on the target line, if and only
  if~$g(x_{i_1},\dots,x_{i_k})$ evaluates to true.
  All other variables remain unchanged.
  If the definition of~$g$ is obvious from the context, it can be omitted from
  the notation~$\T_g$.
\end{definition}
\begin{definition}[Toffoli gate]
  \emph{Mixed-polarity multiple-control Toffoli~(MPMCT) gates} are a subset of
  the single-target gates in which the control function~$g$ can be represented
  with one product term or~$g =\displaystyle\bigwedge_{k=i}^j x_i^p=1$.
  \emph{Multiple-control Toffoli gates~(MCT)} in turn are a subset from MPMCT
  gates in which the product terms can only consist of positive literals.
\end{definition}
\noindent
Using synthesis algorithms it can easily been shown that any reversible
function~$f\in\bfunc_{n,n}$ can be realized by a reversible circuit with~$n$
lines when using MCT gates.  That is, it is not necessary to add any temporary
lines (ancilla) to realize the circuit. This can be the case if the MCT (or
MPMCT) gates are restricted to a given size, e.g.~three bits.  Note that each
single-target gate can be expressed in terms of a cascade of MPMCT or MCT gates,
which can be obtained from an ESOP or PPRM expression~\cite{Sas93},
respectively.  For drawing circuits, we follow the established conventions of
using the symbol $\oplus$ to denote the target line, solid black circles to
indicate positive controls and white circles to indicate negated controls.
\begin{figure}[t]
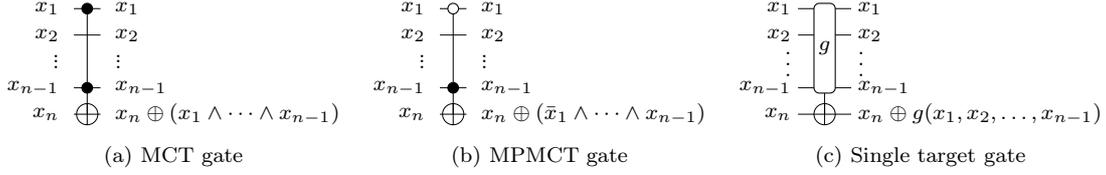

  \centering
    \subfloat[MCT gate\label{fig:mct}]{\tikzmct}\hfill
    \subfloat[MPMCT gate\label{fig:mpmct}]{\tikzmpmct}\hfill
      \subfloat[Single target gate\label{fig:stg}]{\tikzSingleTarget}
    \caption{Reversible gates}
\label{sec:reversible-gates}
\end{figure}
\begin{figure}[t]
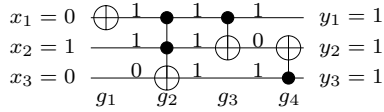

  \centering    
 \tikzrevcircuit
    \caption{Reversible circuit}
\label{fig:revcircuit}
\end{figure}

\begin{example}
  Figure~\ref{fig:mct} shows a Toffoli gate with $n$ positive controls,
  Fig.~\ref{fig:mpmct} shows a Toffoli gate with mixed polarity control lines,
  and Fig.~\ref{fig:stg} shows the diagrammatic representation of a
  single-target gate based on Feynman's notation.  Fig.~\ref{fig:revcircuit}
  shows different Toffoli gates in a cascade forming a reversible circuit.  The
  annotated values demonstrate the computation of the gate for a given input
  assignment.
\end{example}

\begin{figure*}[t]
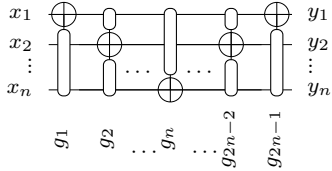

  \centering
  \tikzVshape
  \caption{Synthesis based on Young subgroups}
  \label{fig:young-subgroup-synthesis}
\end{figure*}


\section{Upper bound for Single-Target Gate Circuits}
\begin{theorem}
  \label{thm:upper-bound}
  Let~$f\in\bfunc_{n,n}$ be reversible and $x$ a variable in~$f$.  Then, $f$ can
  be decomposed as~$f=g_2\circ f'\circ g_1$ into three
  functions~$g_1,f',g_2\in\bfunc_{n,n}$ such that~$f'$ is a reversible function
  that does not change~$x$, and~$g_1$ and~$g_2$ can each be realized as a
  single-target gate that acts on~$x$.
\end{theorem}

\begin{proof}
  Reversible functions of $n$ variables are isomorphic to the symmetric group
  $S_{2^n}$.  Consequently, $f$ corresponds to an element~$a\in S_{2^n}$.  The
  element $a$ can be decomposed as~$a=h_2vh_1$, where both~$h_1$ and~$h_2$ are
  members of the Young subgroup~$S_{2^{n-1}}^2$ and $v$ is a member
  of~$S_2^{2^{n-1}}$~\cite{RVS05}.  From~$h_1$, $h_2$, and~$v$ one can
  derive~$g_1$, $g_2$, and~$f'$~\cite{VR08}.
\end{proof}

\begin{corollary}
  \label{cor:upper-bound}
  Each reversible function~$f\in\bfunc_{n,n}$ can be implemented as a reversible
  circuit with at most~$2n-1$ single-target gates.
\end{corollary}

\begin{proof}
  When applying Theorem~\ref{thm:upper-bound} to all variables in an iterative
  manner, $f'$ will be the identity function after at most~$n$ steps and at
  most~$2n$ gates have been collected.  Since $f'$ is the identity function in
  the last step, the last two gates can be combined into one single-target gate.
\end{proof}

\sn
A truth table based algorithm that makes use of the results of
Theorem~\ref{thm:upper-bound} and Corollary~\ref{cor:upper-bound} has been
presented in~\cite{VR08}.  Since the variables are selected in a decreasing
order, the target lines of the resulting single-target gates are aligned on a
V-shape (cf.~Fig.~\ref{fig:young-subgroup-synthesis}).
\section{Lower Bound for Toffoli Gate Circuits}
Reversible circuits with $n$ inputs that consist of only one MCT gate can
represent $n\cdot 2^{n-1}$ reversible functions.  There are $n$ possible
positions to fix the target line and then $n-1$ positions remain to either put
or not put a control line.  If one has two Toffoli gates one can represent at
most $(n\cdot 2^{n-1})(n\cdot 2^{n-1})$ functions.  The actual number is
smaller, since for some circuits the order of gates does not matter or both
gates are equal which corresponds to the empty circuit.  In general, one can
represent at most~$(n\cdot2^{n-1})^k$ reversible functions with a circuit that
consists of $k$ Toffoli gates.  Since there are $2^n!$ reversible functions one
can derive that there is at least one function that requires
\begin{equation}
  \label{eq:lower-bound}
  \left\lceil
    \frac{\log(2^n!)}{\log(n2^{n-1})}
  \right\rceil
\end{equation}
gates.

\begin{theorem}
  \label{thm:lower-bound}
  There exist a reversible function which smallest circuit realization requires
  an exponential number of Toffoli gates.
\end{theorem}

\begin{proof}
  We show that there exist a constant $c$ such that $\left\lceil
    \frac{\log(2^n!)}{\log(n2^{n-1})} \right\rceil \ge c\cdot2^n$.  We have
\[
  \left\lceil \frac{\log(2^n!)}{\log(n2^{n-1})} \right\rceil
  \ge
  \frac{\log(2^n!)}{\log(n2^{n-1})}
  =
  \frac{\log_2(2^n!)}{\log_2(n2^{n-1})}
  \ge c\cdot2^n
\]
which can be rewritten to
\[
  \log_2(2^n!)\ge c\cdot2^n\log_2(n2^{n-1})=c\cdot2^n\left(\log_2n+(n-1)\right).
\]
Since $\left(\log_2n+(n-1)\right)<2n$ we are left to prove that~$\log_2(2^n!)\ge
c_0\cdot 2^nn$ for some constant~$c_0$, which we do using induction on~$n$.
From the base case we obtain~$c_0=\frac{1}{2}$.  Assume for some $n$ we have
$\log_2(2^n!)\ge\frac{1}{2}\cdot2^nn$, then in the induction step we get
\begin{equation}
  \label{eq:first-step}
  \log_2(2^{n+1}!)=\sum_{k=1}^{2^{n+1}}\log_2k
  =\sum_{k=1}^{2^n}\log_2k + \sum_{k=1}^{2^n}\log_2(k+2^n)
  =\log_2(2^n!)+ \sum_{k=1}^{2^n}\log_2(k+2^n).
\end{equation}
We will now derive a lower bound for the second term in the last expression
of~\eqref{eq:first-step}. We have
\[
  \sum_{k=1}^{2^n}\left(\log_2(k+2^n)-\log_2k\right)
  =
  \sum_{k=1}^{2^n}\log_2\frac{k+2^n}{k}
  =
  \sum_{k=1}^{2^n}\log_2\left(1+\frac{2^n}{k}\right)
  \ge
  \sum_{k=1}^{2^n}1=2^n
\]
from which we derive
\[
  \sum_{k=1}^{2^n}\log_2(k+2^n)
  \ge
  \sum_{k=1}^{2^n}\log_2k+2^n
  =
  \log_2(2^n!)+2^n.
\]
Plugging this into~\eqref{eq:first-step} we get
\[
  \log_2(2^{n+1}!)
  \ge
  \log_2(2^n!)+\log_2(2^n!)+2^n
  \ge
  2\cdot\left(\frac{1}{2}\cdot 2^nn\right)+2^n
  =\frac{1}{2}\cdot2^{n+1}(n+1).\qedhere
\]
\end{proof}

\sn The proof for Theorem~\ref{thm:lower-bound} in~\cite{MDM05} uses mEXOR gates
as underlying gate library, which generalize MCT gates to have more than one
target line.  The proof can be carried out analogously for MPMCT gates.

\section{Framework for Circuit Complexity Analysis}
\label{sec:framework}
\begin{figure}[t]
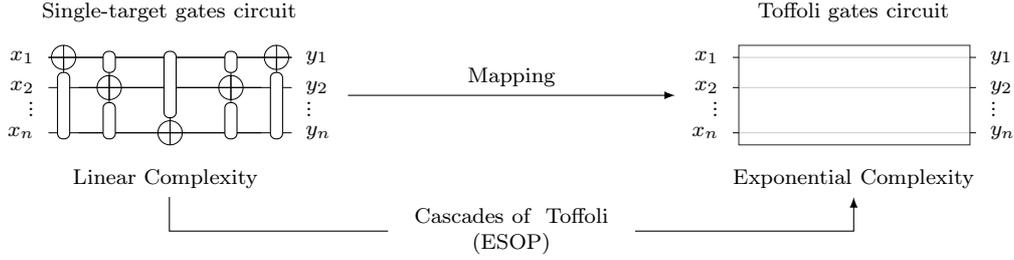

  \centering
   \tikzdiagram
  \caption{Reversible circuit complexity}
  \label{fig:revcomp}
\end{figure}
Both bounds that have been presented in the previous section are used to
motivate a framework for reversible circuit complexity analysis that is
illustrated in Fig.~\ref{fig:revcomp}.  If only considering single-target gates,
one knows that at most linear many gates are required.  When being restricted to
Toffoli gates, there are functions that need at least an exponential number of
gates.  Since single-target gates can be translated to cascades of Toffoli
gates, there must be control functions which require exponential many product
terms when being decomposed into ESOP expressions.

However, when restricting the ESOP mapping interesting cases arise.  Let us
first consider that we only allow single-target target gates which control
function can be mapped to ESOP expressions that have a constant number of
product terms, e.g.~$1$ or~$2$.  Then the resulting Toffoli circuits are of
linear size.  The most interesting and important question is how to determine
which class of reversible functions can be represented using these circuits when
applying this restriction.  We experienced this to be a difficult research
problem in our investigations on this topic.  The same idea can be extended to
other cases.  If e.g.~we allow those single-target gates which control functions
can be mapped to linear size ESOP expressions, one obtains Toffoli circuits of
quadratic size.

\section{Application: ``Better than Optimal'' Embedding}
\label{sec:application}
In this section we consider half of the V-shaped circuit that has been
considered in the previous section.  That is, we consider reversible circuits
with $n$ variables and $n$ single-target gates that have their target lines in
subsequent order from the top line to the bottom line.  Also, no two
single-target gates have the target line in common.  It can easily be seen that
at most $\left(2^{2^{n-1}}\right)^n$ different reversible functions can be
realized with such circuits.

We will now show that in fact exactly $\left(2^{2^{n-1}}\right)^n$ different
functions can be realized, which implies that these circuits are a canonical
representation for this subset of reversible functions.

\begin{lemma}
  \label{lem:canonical-k}
  Reversible circuits with $k\ge n$ lines that have $n$ single-target gates with
  pairwise different target lines in increasing order on lines $1$ to $n$ can
  realize exactly $\left(2^{2^{k-1}}\right)^n$ reversible functions.
\end{lemma}

\begin{proof}
  We prove this using induction on~$n$.  The base case is simple since each
  single-target gate realizes a different function and since the target line is
  fixed on the first line, there are $2^{2^{k-1}}$ possibilities to choose the
  control function.  Assuming the claim holds for all circuits with up to $n$
  gates, we consider a circuit that has $n+1$ gates.  The subcircuit $C'$
  consisting of the first $n$ gates realizes $\left(2^{2^{k-1}}\right)^n$
  functions due to the induction hypothesis.  Since the $(n+1)$\textit{th} gate
  has its target line on a line that has not been used as target line in $C'$
  and since there are no two gates that realize the same function, the statement
  follows.
\end{proof}

\begin{corollary}
  Reversible circuits with $n$ lines that have $n$ single-target gates with
  their targets being on increasing lines (from the top to the bottom) can
  realize exactly $\left(2^{2^{n-1}}\right)^n$ reversible functions.
\end{corollary}

\noindent
There are also $\left(2^{2^{n-1}}\right)^n$ Boolean multiple-output functions
in~$\bfunc_{n-1,n}$.  When realizing these functions as reversible circuits one
needs to embed them first into reversible functions, which will have up to
$2n-1$ variables.  The additional variables are required to ensure that the
function is bijective.  However, since there is a 1-to-1 correspondence between
the number of half V-shaped circuits on $n$ lines and the number of functions in
$\bfunc_{n-1,n}$, a ``better than optimal'' embedding is possible because in the
conventional case there are functions that require at least $2n-1$ lines
(e.g.~the constant functions).  One only needs to define a mapping function from
the multiple-output function to the reversible one as well as an interpretation
function for the computed outputs.  Based on Lemma~\ref{lem:canonical-k} this
embedding technique can be extended to functions in~$\bfunc_{k-1,n}$ with $k\ge
n$.

\section{Conclusions}
\label{sec:conclusion}
We have motivated a framework for the analysis of complexity of reversible
circuits based on two bounds.  As one first application of this framework we
have presented an idea for a better than optimal embedding.  The research in
this area is still in its infancy so far, however, the discussions started in
this paper provide a starting point to tackle the open problems.  One direction
for future work is to find a way to derive function classes from a subset of
reversible circuits.  It is also open whether these will be function classes
known from the literature or whether new ones for the special case of reversible
functions need to be defined.

\paragraph{Acknowledgments.} We thank Eugenia Rosu for her help with the proof
to Theorem~\ref{thm:lower-bound}.  We also thank Amatulwaseh Hayat for many
interesting discussions.

\end{document}